\documentclass{article}
\usepackage[margin=.75in]{geometry}

\usepackage{cite}
\usepackage{amsmath}
\usepackage{amssymb}
\usepackage{amsfonts}
\usepackage{algorithmic}
\usepackage{graphicx}
\usepackage{textcomp}
\usepackage{xcolor}
\usepackage{bm}
\usepackage{hyperref}
\hypersetup{
  colorlinks,
  linkcolor={green!60!black},
  citecolor={red!70!black},
  urlcolor={blue!70!black}
}

\usepackage{xurl}

\usepackage{amsmath}
\usepackage{amssymb}
\usepackage{amsfonts}
\usepackage{amsthm}
\usepackage{array}
\usepackage{enumerate}

\newif\ifone
\newif\iftwo

\usepackage{bm}
\usepackage[caption=false,font=normalsize,labelfont=sf,textfont=sf]{subfig}
\usepackage{textcomp}

\usepackage{verbatim}
\usepackage{graphicx}

\usepackage{silence}
\usepackage{float}

\usepackage{soul}
\usepackage{color}
\usepackage{xcolor}
\usepackage{booktabs}

\usepackage{balance}

\usepackage[utf8]{inputenc} 
\usepackage[T1]{fontenc}
\usepackage{url}
\usepackage{ifthen}
\usepackage{cite}

\allowdisplaybreaks

\newtheorem{theorem}{Theorem}
\newtheorem{lemma}[theorem]{Lemma}
\newtheorem{corollary}[theorem]{Corollary}

\newtheorem{definition}{Definition} 
\newtheorem{remark}{Remark}

\newtheorem*{openproblem}{Open Problem}

\newcommand{\smat}{\ \ }

\newlength{\spacelen}

\newcommand{\rk}[1]{\textnormal{rank}\left(#1\right)}
\newcommand{\nz}[1]{\textnormal{nz}\left(#1\right)}
\newcommand{\Angle}[1]{\left\langle#1\right\rangle}

\newcommand{\ff}{\mathbb{F}}
\newcommand{\Fq}{\mathbb{F}_q}

\newcommand{\bO}{\mathbf{O}}

\newcommand{\cC}{\mathcal{C}}

\newcommand{\cM}{\mathcal{M}}
\newcommand{\cN}{\mathcal{N}}

\newcommand{\cR}{\mathcal{R}}
\newcommand{\cS}{\mathcal{S}}
\newcommand{\cT}{\mathcal{T}}

\newcommand{\mA}{A}
\newcommand{\mB}{B}

\newcommand{\mL}{L}

\newcommand{\mO}{O}

\newcommand{\mR}{R}

\newcommand{\ve}{\bm{e}}

\newcommand{\vu}{\bm{u}}
\newcommand{\vv}{\bm{v}}

\newcommand{\vlmd}{\bm{\lambda}}

\newcommand{\sG}{\mathcal{G}}

\newcommand{\sM}{\mathcal{M}}

\newcommand{\sR}{\mathcal{R}}
\newcommand{\sS}{\mathcal{S}}
\newcommand{\sT}{\mathcal{T}}

\DeclareMathOperator{\tB}{BW}
\DeclareMathOperator{\tO}{IO}

\newcommand{\ceil}[1]{\left\lceil #1 \right\rceil}
\newcommand{\floor}[1]{\left\lfloor #1 \right\rfloor}

\DeclareMathOperator{\MOD}{mod}

\newcommand{\RK}[1]{\textnormal{rank} #1}

\begin{document}

\title{Optimal Repair of $(k+2, k, 2)$ MDS Array Codes
  \thanks{This work is partially supported by
    the National Key R\&D Program of China (Grant No. 2021YFA1001000),
    the National Natural Science Foundation of China (Grant No. 12231014), and
    the Postdoctoral Fellowship Program and China Postdoctoral Science Foundation (Grant No. BX20250065).

    Zihao Zhang and Guodong Li are with the Key Laboratory of Cryptologic Technology and Information Security, Ministry of Education, and the School of
    Cyber Science and Technology, Shandong University, Qingdao, Shandong 266237,
    China (e-mail: zihaozhang@mail.sdu.edu.cn, guodongli@sdu.edu.cn). Sihuang Hu is
    with the Key Laboratory of Cryptologic Technology and Information Security,
    Ministry of Education, and the School of Cyber Science and Technology, Shandong
    University, Qingdao, Shandong 266237, China, and also with the Quan Cheng
    Laboratory, Jinan 250103, China (e-mail: husihuang@sdu.edu.cn).
  }
}
\author{Zihao Zhang, Guodong Li, and Sihuang Hu}
\date{}
\maketitle
\begin{abstract}

    Maximum distance separable (MDS) codes are widely used in distributed storage systems as they provide optimal fault tolerance for a given amount of storage overhead.
    The seminal work of Dimakis~\emph{et al.} first established a lower bound on the repair bandwidth for a single failed node of MDS codes, known as the \emph{cut-set bound}. MDS codes that achieve this bound are called minimum storage regenerating (MSR) codes. Numerous constructions and theoretical analyses of MSR codes reveal that they typically require exponentially large sub-packetization levels, leading to significant disk I/O overhead. To mitigate this issue, many studies explore the trade-offs between the sub-packetization level and repair bandwidth, achieving reduced sub-packetization at the cost of suboptimal repair bandwidth. Despite these advances, the fundamental question of determining the minimum repair bandwidth for a single failure of MDS codes with fixed sub-packetization remains open.

    In this paper, we address this challenge for the case of two parity nodes ($n-k=2$) and sub-packetization $\ell=2$. Under these parameters, we establish a correspondence between repair schemes and point sets on the projective line \(\mathbb{P}^1\), and then derive a lower bound on repair bandwidth utilizing the sharply 3-transitive action of \(\text{PGL}_2(\Fq)\). Furthermore, we extend this lower bound to the repair I/O, and construct two classes of explicit MDS array codes that achieve these bounds, offering practical code designs with provable repair efficiency.
\end{abstract}

\section{Introduction}\label{sec:intro}

Erasure codes are widely adopted in distributed storage systems as they provide better storage efficiency for the same level of fault tolerance compared to replication~\cite{weatherspoon02}. Consider an $(n,k)$ erasure-coded distributed storage system. For a file of size $\cM$ bits to be stored, it is first divided into $k$ \emph{data packets}, denoted by $C_1, C_2, \dots, C_k$, each of size $\cM/k$. Then, using the erasure code, $r = n - k$ \emph{parity packets} are generated, denoted by $C_{k+1}, C_{k+2}, \dots, C_n$, each also of size $\cM/k$. Finally, these $n$ data and parity packets are stored on $n$ distinct storage \emph{nodes}. In distributed system literature, these $n$ data and parity packets are called a \emph{stripe}. If a packet in a stripe is lost, it can be recovered by accessing some remaining data or parity packets.

\noindent\textbf{Reed-Solomon (RS) codes.}
Reed-Solomon (RS) codes are the most prevalent class of erasure codes. They are widely deployed in distributed storage systems (e.g., Google~\cite{ford10}, Facebook~\cite{muralidhar14}) because they support general $(n,k)$ (where $n>k$) and achieve the optimal trade-off between storage overhead and fault tolerance. More precisely, a key advantage is their maximum distance separable (MDS) property, i.e., in an $(n, k)$ RS-coded stripe, any $k$ packets suffice to reconstruct the stripe.

Traditional RS code repair strategies incur high \emph{repair bandwidth}---the total amount of data transmitted over the network to recover failed nodes. Specifically, repairing up to $r$ failed packets in an RS-coded stripe requires transferring $k$ packets, which is equivalent to the size of the original file.

Fig.~\ref{pic:64} illustrates a $(6,4)$ RS-coded stripe. In this example, if any single packet fails, the repair process requires transferring all four remaining packets, resulting in a repair bandwidth of $\cM$ bits.

\begin{figure}[ht]
    \centering
\includegraphics[width=0.6\linewidth]{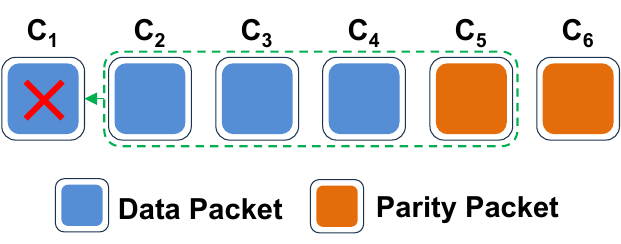}
    \caption{An $(n,k)=(6,4)$ RS-coded stripe.}
    \label{pic:64}
\end{figure}

\noindent\textbf{MSR codes.}
To reduce the repair bandwidth of RS codes, Dimakis~\emph{et al.} introduced the concept of minimum storage regenerating (MSR) codes, which minimize the repair bandwidth for a single failed node while maintaining the MDS property~\cite{dimakis10}. The minimum repair bandwidth of MDS codes is known as the \emph{cut-set bound}. The key idea behind MSR codes is to further partition each packet into $\ell$ \emph{sub-packets}, allowing for a more granular encoding process. The number of sub-packets in each packet $\ell$ is referred to as the \emph{sub-packetization}. MSR codes achieve reduced repair bandwidth by allowing each helper node to transmit only a portion of its sub-packets during the repair process.

Fig.~\ref{pic:648} illustrates a $(6,4,8)$ MSR-coded stripe, where each packet is divided into $\ell=8$ sub-packets. In this example, if any single packet is lost, the repair process requires transferring four sub-packets from each of the remaining five packets; the resulting repair bandwidth is twenty sub-packets. Since each packet is divided into eight sub-packets, the repair bandwidth is equivalent to 2.5 packets, which significantly saves $37.5\%$ repair bandwidth compared to the (6, 4) RS-coded stripe in Fig.~\ref{pic:64}.

\begin{figure}[ht]
    \centering
    \includegraphics[width=0.6\linewidth]{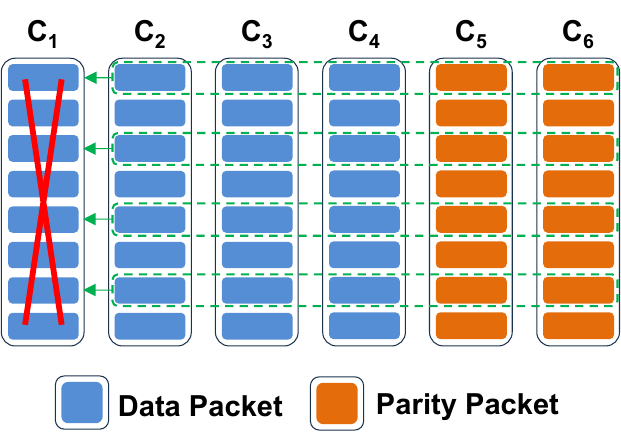}
    \caption{An $(n,k,\ell)=(6,4,8)$ MSR-coded stripe.}
    \label{pic:648}
\end{figure}

Since the introduction of MSR codes, numerous MSR code constructions have been proposed. Product-Matrix (PM) MSR codes~\cite{rashmi11} and PM-RBT~\cite{rashmi15} proposed by Rashmi \emph{et al.} are the first explicit MSR codes, but they have high storage overhead ($n \geq 2k-1$). Functional MSR (F-MSR) codes~\cite{hu12} support $n-k=2$ and $n-k=3$, yet they are non-systematic, meaning original data cannot be directly read from the stripes. Butterfly codes~\cite{pamiesJuarez16} are the first binary MSR codes, but they only support $n-k=2$. Subsequent MSR constructions supporting general $n$ and $k$ emerged, such as Ye-Barg codes~\cite{ye17,ye17a}, the generic transformation from MDS codes to MSR codes~\cite{li18}, and Clay codes~\cite{vajha18}. Notably, Clay codes~\cite{vajha18} have been deployed in the distributed storage system Ceph~\cite{weil06}. The above constructions only support $n-1$ helper nodes. Recently, Li~\emph{et al.} proposed a universal MSR code construction with the smallest sub-packetization that supports any number of helper nodes~\cite{li24}. A comprehensive review of MSR constructions can be found in the monograph~\cite{ramkumar22}.

Among known MSR codes, all systematic code constructions with low storage overhead require exponential sub-packetization levels $\ell = \exp(O(n))$. For example, The sub-packetization of an $(n,k)$ Clay code is $\ell = (n-k)^{\lceil n/(n-k) \rceil}$. Concurrently, multiple theoretical lower bounds~\cite{alrabiah19,balaji18,balaji22} demonstrate that exponential sub-packetization is unavoidable for high-rate systematic MSR codes.

Exponential sub-packetization severely hinders the practical adoption of MSR codes in distributed storage systems. While fine-grained partitioning of packets reduces repair bandwidth, it incurs high I/O overheads. Specifically, when the data and parity packets are divided into numerous sub-packets, the repair process of single failures requires helper nodes to transmit partial sub-packets to failed nodes, generating extensive non-contiguous disk I/O. For systems where disk I/O efficiency is the bottleneck, this significantly degrades performance.

As shown in Fig.~\ref{pic:648}, to repair the first packet in a $(6, 4, 8)$ MSR-coded stripe, each of the remaining five packets transmits four sub-packets, resulting in a total of 20 non-contiguous I/O operations on the stripe. This is a significant overhead compared to the four I/O operations required for repairing a single packet in the RS-coded stripe in Fig.~\ref{pic:64}.

\begin{figure}[ht]
    \centering
    \includegraphics[width=0.6\linewidth]{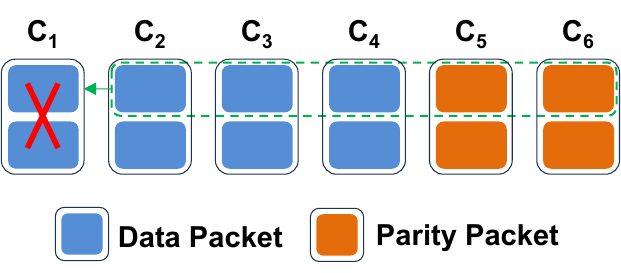}
    \caption{An $(n,k,\ell)=(6,4,2)$ MDS-coded stripe.}
    \label{pic:642}
\end{figure}

\noindent\textbf{MDS array codes with small sub-packetization.}
To reduce disk I/O overhead while maintaining low repair bandwidth, many works consider the trade-off between the exponential sub-packetization level of MSR codes and optimal repair bandwidth, focusing on MDS array codes with small sub-packetization. Such schemes preserve the MDS property like RS codes, while introducing reduced sub-packetization to achieve moderate repair bandwidth between that of RS codes and MSR codes.

As shown in Fig.~\ref{pic:642}, a $(6,4,2)$ MDS-coded stripe divides each packet into two sub-packets. In this example, if any single packet fails, the repair process requires transferring five (or six) sub-packets from the remaining five packets, resulting in a saving of 25-37.5\% repair bandwidth across all the data and parity packets compared to the RS-coded stripe in Fig.~\ref{pic:64}. Compared to the MSR-coded stripe in Fig.~\ref{pic:648}, this coding scheme achieves a suboptimal repair bandwidth without incurring non-contiguous disk I/O within one stripe.

The piggybacking framework~\cite{rashmi17} pioneered the construction of MDS codes with small sub-packetization and low repair bandwidth. This framework inherently retains the MDS property of RS codes while supporting arbitrary sub-packetization. Hitchhiker codes~\cite{rashmi14}, a specialized piggybacking design with $\ell = 2$, have been implemented in HDFS, reducing data packet repair bandwidth by 25\%--45\% compared to RS codes, though parity packet repair bandwidth remains unchanged. HashTag codes~\cite{kralevska18-TBG} similarly reduce repair bandwidth for data blocks. HashTag+ codes~\cite{kralevska18-DASC} further optimize repair bandwidth for both data and parity blocks. 
The Elastic Transformation~\cite{tang23} also converts RS codes into repair-efficient codes and supports configurable sub-packetization. Additionally, $\varepsilon$-MSR codes~\cite{rawat17,ramkumar25} demonstrate that for sufficiently large $n$, logarithmic sub-packetization suffices to achieve $(1+\varepsilon)$-times the optimal repair bandwidth. Other notable constructions with near-optimal bandwidth are proposed in~\cite{li21, li23MDS,cheng26,kadekodi23}.

Despite significant progress in constructing such MDS codes, the fundamental lower bound on repair bandwidth subject to a prescribed sub-packetization level remains largely unexplored.
The following open problem posed by Ramkumar~\emph{et al.} in their survey~\cite{ramkumar22} addresses this fundamental question.

\begin{openproblem}\cite[Open Problem 9]{ramkumar22}\label{op:ramkumar}
    Characterize the tradeoff between repair bandwidth, sub-packetization level, and field size for the general class of vector MDS codes.
\end{openproblem}

In~\cite{guruswami17}, Guruswami and Wootters first studied the repair problem of RS codes. By viewing each element of $\ff_{q^\ell}$ as an $\ell$-length vector over $\Fq$, they proposed repair schemes that outperform the naive approach and proved a lower bound on repair bandwidth that holds for any scalar MDS codes. 
Here, the degree of field extension $\ell$ is the sub-packetization.
Dau and Milenkovic~\cite{dau17} subsequently tightened this bound for RS codes.

More recently, attention has turned to the repair I/O of scalar MDS codes over $\ff_{q^\ell}$, defined as the total number of \(\Fq\)-elements accessed during repair.
Dau~\emph{et al.} \cite{dau18} provided the first nontrivial repair I/O lower bound for full-length RS codes over fields of characteristic $2$ with $2$ parity nodes. 
Subsequent works~\cite{li19io, liu25, liu24} have extended and refined these bounds for broader classes of RS codes.

\subsection{Our Contributions}

We investigate the special case of Open Problem~\ref{op:ramkumar} under fixed parameters $r = n-k = 2$ and $\ell = 2$.
Wu~\emph{et al.} addressed a variant of Open Problem~\ref{op:ramkumar} also in the special case of $n-k=2$ and $\ell = 2$; their work requires the codes to be degraded-repair-friendly~\cite{wu21}.
We impose no degraded-repair-friendly restriction.
The redundancy $r = 2$ is widely deployed in practice (e.g., RAID-6~\cite{blaum95} and Tencent Ultra-Cold Storage~\cite{intel24}.)
The sub-packetization $\ell = 2$ further ensures that, within one stripe, every packet is accessed contiguously during repair, eliminating non-contiguous I/O.

On the one hand, we give the theoretical limits on the repair efficiency of $(n = k+2, k, \ell = 2)$ MDS array code. Firstly, we consider the minimum repair bandwidth (denoted by $\beta_i$) of each packet $C_i$ ($1\le i\le n$) in one ($n = k+2, k, \ell=2$) MDS-coded stripe, where the repair bandwidth is the number of sub-packets transferred during the repair of $C_i$. Then, we consider the minimum repair I/O (denoted by $\gamma_i$) of each packet $C_i$ ($1\le i\le n$) in one ($n = k+2, k, \ell=2$) MDS-coded stripe, where the repair I/O is the number of sub-packets accessed on the helper packets during the repair of $C_i$. We provide the following lower bounds on the minimum repair bandwidth and the minimum repair I/O.
\begin{theorem}
    \label{thm:avg-min-bw}
    Let $\cC$ be a $(k+2, k,2)$ MDS array code and $\beta_i$ the minimal repair bandwidth for each packet $C_i$ (where $1\le i\le n$).
    Then the {\bf avg-min repair bandwidth} $\bar\beta(\cC) := \frac{1}{n} \sum_{i=1}^{n}\beta_i$ satisfies that
    \begin{equation*}
        \label{eq:avg-min-bw}
        \bar\beta(\cC) \geq \frac{5k}{4}.
    \end{equation*}
\end{theorem}
\begin{corollary}
    \label{coro:max-min-bw}
    Let $\cC$ be a $(k+2, k,2)$ MDS array code and $\beta_i$ the minimal repair bandwidth for each packet $C_i$ (where $1\le i\le n$).
    Then the {\bf max-min repair bandwidth} $\beta(\cC) := \max_{i\in[n]}\{\beta_i\}$ satisfies that
    \begin{equation*}
        \label{eq:max-min-bw}
         \beta(\cC) \geq \left\lceil\frac{5k}{4}\right\rceil.
    \end{equation*}
\end{corollary}
\begin{theorem}
    \label{thm:avg-min-io}
    Let $\cC$ be a $(k+2, k,2)$ MDS array code and $\gamma_i$ the minimal repair bandwidth for each packet $C_i$ (where $1\le i\le n$).
    Then the {\bf avg-min repair IO} $\bar\gamma(\cC) := \frac{1}{n}\sum_{i=1}^n \gamma_i$ satisfies that
    \begin{equation*}
        \label{eq:avg-min-io}
        \bar\gamma(\cC)  \geq \frac{4k+1}{3}.
    \end{equation*}
\end{theorem}
\begin{corollary}
\label{coro:max-min-io}
    Let $\cC$ be a $(k+2, k,2)$ MDS array code and $\gamma_i$ the minimal repair bandwidth for each packet $C_i$ (where $1\le i\le n$).
    Then the {\bf max-min repair IO} $\gamma(\cC) := \max_{i\in[n]}\{\gamma_i\}$ satisfies that
    \begin{equation*}
        \label{eq:max-min-io}
        \gamma(\cC) \geq \left\lceil\frac{4k+1}{3}\right\rceil.
    \end{equation*}
\end{corollary}

On the other hand, we construct two classes of $(n = k +2, k, \ell = 2)$ MDS array codes that achieve the lower bound on repair bandwidth (Theorem~\ref{thm:avg-min-bw}) and the lower bound on repair I/O (Theorem~\ref{thm:avg-min-io}), respectively, which shows that those lower bounds are all tight.

The rest of the paper proceeds as follows. Section~\ref{sec:linear-repair} describes the linear repair process of MDS array codes. Section~\ref{sec:bounds} characterizes the theoretical lower bounds on repair bandwidth and repair I/O for $(n = k+2, k, \ell = 2)$ MDS array codes. Section~\ref{sec:construction} constructs two classes of such codes that achieve the proposed lower bounds. Section~\ref{sec:conclusion} concludes the paper.

\section{MDS Codes and Their Repair Schemes}\label{sec:linear-repair}

For convenience, we assume that each storage node contains exactly one packet, with each packet represented as a column vector of length $\ell$ over the finite field $\mathbb{F}_q$. As a result, each sub-packet can be seen as a symbol from $\mathbb{F}_q$. One stripe corresponds to a codeword of an MDS code. In addition, we use the notation $[n]$ to denote the set $\{1, 2, \dots, n\}$.

\noindent\textbf{MDS codes.}
An $(n, k, \ell)$ linear MDS array code $\cC$ over the finite field $\Fq$ comprises $n$ nodes (denoted by $C_1, C_2, \dots, C_n$),
where some $k$ nodes are data nodes and the rest $r=n-k$ nodes are parity nodes, each node is a column vector of length $\ell$ over $\Fq$. In addition, any $r$ out of the $n$ nodes can be recovered from the remaining $k$ nodes.

The relationship among the $n$ nodes of an $(n, k, \ell)$ MDS array code $\cC$ can be described by the following parity check equations:
\begin{equation}\label{eq:pc}
    H_1C_1+H_2C_2+\cdots + H_nC_n = \bm 0,
\end{equation}
where each $H_i$ ($i\in[n]$) is an $r\ell \times \ell$ parity check sub-matrix over $\Fq$.
Then, the $r$ nodes $C_{i_1}, C_{i_2}, \dots, C_{i_r}$ can be recovered from the remaining $k$ nodes if and only if the square matrix
$[H_{i_1} \smat H_{i_2} \smat \dots \smat H_{i_r}]$ is invertible.

\noindent\textbf{Repair scheme driven by a repair matrix.}
For any node $C_i$ of an $(n, k, \ell)$ MDS code $\cC$, any linear repair process of $C_i$ can be described by an $\ell\times r\ell$ \emph{repair matrix} $M$ over $\Fq$. Conversely, any $\ell\times r\ell$ matrix $M$ satisfying specific conditions can characterize a linear repair process for an MDS code $\cC$.

Let $M$ be an $\ell \times r\ell$ matrix over $\Fq$. Multiplying $M$ with the parity check equations~\eqref{eq:pc}, we obtain the following \emph{repair equations}:
\begin{equation}
    \label{eq:repair}
    M H_1 C_1 + M H_2 C_2 + \cdots + M H_n C_n = \bm 0.
\end{equation}
If the square matrix $M H_i$ is invertible for some $i$ (where $1\le i\le n$), then the repair equations~\eqref{eq:repair} can be rewritten as
\[
    C_i=-(M H_i)^{-1}\sum_{j\in[n]\setminus\{i\}}M H_j C_j,
\]
which means $C_i$ can be computed from the $n-1$ vectors $M H_j C_j, j\in[n]\setminus\{i\}$.

\noindent\textbf{Repair bandwidth and repair I/O.}
Now we analyze the repair bandwidth and repair I/O of the repair process above, driven by the repair matrix $M$. Suppose that the $\ell\times \ell $ matrix $MH_i$ is invertible for some $i\in[n]$. Then, each helper node $C_j$ ($j\in[n]\setminus\{i\}$) needs to transmit the vector  $M H_j C_j$ to the failed node $C_i$. For each $j\in[n]$, we denote the rank of the $\ell\times \ell$ matrix $M H_j$ as $\rk{M H_j}$, and the number of nonzero columns in $M H_j$ as $\nz{M H_j}$. Then, the square matrix $MH_j$ can be decomposed as
\begin{equation*}
    \label{eq:decompose}
    M H_j = \mL_j \cdot \mR_j,
\end{equation*}
where $\mL_j$ is an $\ell\times \rk{M H_j}$ matrix and $\mR_j$ is an $\rk{M H_j}\times \ell$ matrix. We can obtain $\mL_j$ by selecting $\rk{MH_j}$ independent columns of $MH_j$, and obtain $\mR_j$ from the linear combination coefficients for the columns of $\mL_j$ to generate the columns in $MH_j$. We have $\rk{\mR_j} = \rk{MH_j}$, $\nz{\mR_j} = \nz{MH_j}$, and the repair equations~\eqref{eq:repair} can be further rewritten as
\begin{equation*}
    \label{repair_eqn}
    C_i = -(M H_i)^{-1} \cdot \sum_{j\in[n]\setminus\{i\}} \mL_j (\mR_jC_j).
\end{equation*}
Since $\mL_j$, $j\in [n]\setminus\{i\}$, can be calculated from $M$ and the parity check sub-matrix $H_j$, it is sufficient to transmit the $n-1$ vectors $\mR_j C_j$, $j\in [n]\setminus\{i\}$ to the failed node $C_i$ for the repair.

Note that each $\mR_j C_j$ is a vector of length $\rk{MH_j}$ over $\Fq$, and transmitting $\mR_j C_j$ requires to read $\nz{\mR_j}$ symbols from $C_j$. Therefore, during the repair process of $C_i$ driven by the repair matrix $M$, the repair bandwidth, i.e., the total number of symbols transmitted from the helper nodes to the failed node, is
\begin{equation}\label{eq:BW}
    \tB(M) =  \sum_{j\in[n]\setminus\{i\}}\rk{MH_j} = \sum_{j\in[n]}\rk{MH_j} - \ell.
\end{equation}
Similarly, the repair I/O, i.e., the total number of symbols read from the helper nodes, is
\begin{equation}\label{eq:IO}
    \begin{aligned}
        \tO(M) & = \sum_{j\in[n]\setminus\{i\}}\nz{\mR_j} = \sum_{j\in[n]\setminus\{i\}}\nz{M H_j} \\
               & = \sum_{j\in [n]}\nz{MH_j} - \ell.
    \end{aligned}
\end{equation}
Further, the \emph{repair degree}, i.e., the number of helper nodes during the repair, is
\begin{equation*}
    \label{eq:degree}
    d(M) = \left|\{i\in[n]:M H_i\neq \mO \}\right|-1.
\end{equation*}

\begin{lemma}\label{lem:degree}
    For an $(n, k, \ell)$ MDS code $\cC$ and a nonzero matrix $M$ of size $\ell\times r\ell$, we have $d(M)\geq k$.
\end{lemma}
\begin{proof}
    If $d(M) \leq k-1$, then there exist $r$ distinct $a_1, \dots, a_r\in[n]$ such that $M\cdot[H_{a_1}, H_{a_2}, \cdots,  H_{a_r}]=\bO$. 
    By the MDS property, $[H_{a_1}, H_{a_2},\cdots, H_{a_r}]$ is invertible, and hence $M$ is a zero matrix, which leads to a contradiction.
\end{proof}

We denote the index set of nodes which can be repaired by an $\ell\times r\ell$ matrix $M$ as
\begin{equation}\label{eq:RM}
    \cR(M):=\{i\in[n]:\ \rk{M H_i}=\ell\}.
\end{equation}
Further, for each node $C_i$, $i\in [n]$, we denote the matrix set which can repair $C_i$ as
\begin{equation}\label{eq:Mi}
    \cM_i:=\left\{M\in\Fq^{\ell\times r\ell}:i\in\mathcal{R}(M)\right\}.
\end{equation}

Based on the analysis above, we can represent
the avg-min repair bandwidth $\bar \beta(\cC)$, the max-min repair bandwidth $\beta(\cC)$,
the avg-min repair I/O $\bar \gamma(\cC)$ and the max-min repair I/O $\gamma(\cC)$
of an $(n, k, \ell)$ MDS code $\cC$ in terms of 
\begin{align}
    \bar\beta(\cC) &= \frac{1}{n} \sum_{ i\in [n]}\beta_i = \frac{1}{n}\sum_{ i\in [n]}\min_{M\in \cM_i}\{\tB(M)\},\label{def:avg-min-bw}\\
    \beta(\cC) &= \max_{ i\in [n]} \{\beta_i\} =\max_{ i\in [n]}\min_{M\in \cM_i}\{\tB(M)\},\label{def:max-min-bw}\\
    \bar \gamma(\cC) &= \frac{1}{n}\sum_{i\in[n]} \gamma_i =\frac{1}{n}\sum_{i\in[n]}\min_{M\in \cM_i}\{\tO(M)\}, \label{def:avg-min-io}\\   \gamma(\cC) &= \max_{i\in[n]} \{\beta_i\} =\max_{ i\in [n]}\min_{M\in \cM_i}\{\tO(M)\}\label{def:max-min-io}.
\end{align}

\section{Fundamental Limits}\label{sec:bounds}
This section establishes the lower bounds of {\bf 1)} the avg-min repair bandwidth~\eqref{def:avg-min-bw}, and {\bf 2)} the avg-min repair I/O~\eqref{def:avg-min-io} of a $(k+2, k, 2)$ MDS array code $\cC$ over finite field $\Fq$, where $r = n-k = 2, \ell = 2$ and the parity check sub-matrix $H_i$ for node $C_i$ (where $i\in [n]$) is of size $4\times 2$. 
Throughout this section, we work with such codes with a fixed parity-check matrix in systematic form:
\[
H=
\begin{bmatrix}
  A_{1} & \cdots & A_{k} & I_2 & O    \\
  B_{1} & \cdots & B_{k} &  O   & I_2
\end{bmatrix},
\]
where each block $A_i,B_i\in\Fq^{2\times 2}$ is invertible by the MDS property.

\begin{lemma}\label{lem:d=k+1}
For a $(k+2,k,2)$ MDS array code, to determine the minimum repair bandwidth or repair I/O of any node $C_i$, it suffices to consider repair matrices with repair degree $d(M) = k+1$.
\end{lemma}

\begin{proof}
    For any $(n, k, \ell)$ MDS code, the repair degree $d(M)$ must satisfy $d(M) \geq k$. Since $n=k+2$, the repair degree is constrained to $d(M) \in \{k, k+1\}$. Note that trivial repair ($d(M)=k$) yields the fixed bandwidth and I/O of $2k$. Therefore, it is sufficient to demonstrate the existence of a repair matrix $M$ for $C_i$ such that $d(M)=k+1$ and both repair bandwidth and repair I/O $\le 2k$.

    Without loss of generality, assume $C_1$ fails. Consider the matrices:
    $$
    M_1=
    \begin{bmatrix}
        1 & 0 & 0 & 0\\
        0 & 0 & 1 & 0
    \end{bmatrix}
    \quad \text{and} \quad
    M_2=
    \begin{bmatrix}
        1 & 0 & 0 & 0\\
        0 & 0 & 0 & 1
    \end{bmatrix}.
    $$
    Since $W_1$ is invertible, at least one of $M_1H_1$ or $M_2H_1$ is invertible. Thus, there exists a matrix $M \in \{M_1, M_2\}$ capable of repairing $C_1$ with $d(M)=k+1$ and $\tO(M)\leq 2k$. Since $\tB(M)\leq\tO(M)$, its repair bandwidth also no more than $2k$.

    This confirms that the minimum bandwidth and I/O are determined within the case $d(M)=k+1$.
\end{proof}

 Based on Lemma~\ref{lem:d=k+1}, we assume that the repair degree is $k+1$ in the subsequent analysis.

To analyze the minimum repair bandwidth and repair I/O, we define a total order on the power set $2^{[n]}$. It first compares subsets by their cardinality. Then, for subsets of the same cardinality, it compares them by the dictionary order.
\begin{definition}[\bf Total order $\prec$ on the power set $2^{[n]}$]\label{def:order}
    For any two subsets $\sS, \sT\subseteq [n]$, we define their order as follows:
    \begin{itemize}
        \item If $|\sS| < |\sT|$, then $\sS\prec\sT$. 
        \item If $|\sS| = |\sT|$, we suppose that $\sS = \{ s_1, s_2, \dots, s_m \}$ and $T = \{ t_1, t_2, \dots, t_m \}$ where $s_1 < s_2 < \dots < s_m $ and $ t_1 < t_2 < \dots < t_m $.
              Then, we say $\sS \prec \sT$ if and only if there exists an index $i\in[m]$ such that $s_i < t_i$, and for all $j < i$, $s_j = t_j$.
    \end{itemize}
\end{definition}

Now we consider the linear repair process for node $C_i$, $i\in[n]$, and measure the corresponding lower bounds on repair bandwidth and repair I/O.

\subsection{The Lower Bound on Avg-Min Repair Bandwidth}\label{subsection:bw}
In this subsection, since we focus on the repair bandwidth, and blockwise invertible column transformation does not change the repair bandwidth (according to \eqref{eq:BW}), it is convenient to further simplify the parity-check matrix to
\[
H=
\begin{bmatrix}
  I_2 & \cdots & I_2 & I_2 & O    \\
  W_{1} & \cdots & W_{k} &  O   & I_2
\end{bmatrix},
\]
where each $W_i = B_i A_i^{-1}$ is an invertible matrix over $\Fq^{2\times 2}$. 

Recall the definitions of set $\sR(M)$ in~\eqref{eq:RM} and set $\sM_i$ in~\eqref{eq:Mi}.
According to $d(M) = k+1$, we have
\begin{equation*}
    \tB(M) = \sum_{j\in [n]}\rk{MH_j}-2 = k + |\cR(M)|,
\end{equation*}
which means that for a node $C_i, i\in [n]$, and a repair matrix $M\in \sM_i$, the smaller $|\sR(M)|$ is, the smaller the repair bandwidth of $C_i$ is.

For each node $C_i$, let $\cR_i$ be the minimum set in $\{\cR(M):M\in\cM_i\}$ under the total order $\prec$. Let $M_i\in\cM_i$ be a repair matrix such that $\cR(M_i)=\cR_i$.
Then we have
    \begin{itemize}
      \item $M_i$ can repair $C_i$ ($i\in \cR_i$), and
      \item $M_i$ minimizes the repair bandwidth of $C_i$ ($\beta_i = k + |\cR_i|$).
    \end{itemize}
Thus, the avg-min repair bandwidth of code $\cC$ is
\begin{equation}\label{eq:sum_bound}
  \bar\beta(\cC) = \frac{1}{n}\sum_{i\in[n]}BW(M_i)=k+\frac{1}{n}\sum_{i\in[n]}|\cR_i|.
\end{equation}

\begin{remark}\label{rem:i in S}
  For every node $C_i$, $i\in[n]$, we have $i\in\cR_{i}$.
\end{remark}

\begin{lemma}\label{lem:bw-SiSj}
  For any two distinct nodes $C_i$ and $C_j, i,j\in[n]$, if $\cR_{i}\prec \cR_{j}$, then $j\notin \cR_{i}$.
\end{lemma}
\begin{proof}
  Suppose, for the sake of contradiction, that $j\in \cR_{i}$. This implies that the matrix $M_i$ is a valid repair matrix for node $C_j$ ( i.e., $M_i\in \cM_j$). Since $\cR_j$ is minimized under $\prec$, we must have $\cR_{j} \preceq \cR_{i}$. This contradicts the assumption that $\cR_{i}\prec \cR_{j}$. Therefore, we conclude that $j\notin \cR_{i}$.
\end{proof}

For the rest of this section, we will work under the assumption that $\cR_{1}\preceq \cR_{2} \preceq \dots \preceq \cR_{n}$.

\begin{lemma}\label{lem:structure_t}
    There exists an integer $i\in[k]$ such that $\cR_i\prec \cR_{k+1}$. Let
    $
        t := \max\{i\in[k]: \cR_i \prec \cR_{k+1}\}.
    $
    Then, $|\cR_{k+1}|\geq k+1-t$.
\end{lemma}

\begin{proof}
    First, we verify the existence of such an integer $i$. Suppose, for the sake of contradiction, that no such $i$ exists. This would imply $\cR_{1}=\dots=\cR_{k+1}$. Such a condition would result in a repair bandwidth strictly greater than $2k$, contradicting Lemma~\ref{lem:d=k+1}. Thus, the set $\{i\in[k]: \cR_i \prec \cR_{k+1}\}$ is non-empty, and $t$ is well-defined.

    By the definition of $t$, we have $\cR_{t+1} =\dots=\cR_{k+1}$. Remark~\ref{rem:i in S} implies that 
    $
        \{t+1, \dots, k+1\} \subseteq \cR_{k+1}.
    $
    Consequently, $|\cR_{k+1}| \geq |\{t+1, \dots, k+1\}| = k+1-t$.
\end{proof}

The following lemma demonstrates the structure of the repair matrices for the first \(t\) nodes.

\begin{lemma}\label{lem:bw-Fi}
  For each $i\in[t]$, the repair matrix $M_i$ can be decomposed as
  $$ \renewcommand{\arraystretch}{0.8}
  M_{i}=\Lambda
  \begin{bmatrix}
    \vu_i &  \\
     & \vv_i
  \end{bmatrix},$$
  where $\Lambda\in \Fq^{2\times 2}$ is an invertible matrix and $\vu_i,\vv_i\in\Fq^{1\times2}$ are nonzero vectors.
\end{lemma}

\begin{proof}
  Since $\cR_i\preceq \cR_t\prec \cR_{k+1}\preceq \cR_{k+2}$ for each $i\in[t]$, Lemma~\ref{lem:bw-SiSj} implies that $k+1, k+2 \notin \cR_i$. 
  
  By the definition of $\cR_{i}$, we have $\rk{M_{i}H_{j}} \leq 1$ for $j \in \{k+1, k+2\}$. 
  On the other hand, since the repair degree is $d(M_i)=k+1$, the contribution from each helper must be nonzero, implying $\rk{M_{i}H_{j}} = 1$ for $j \in \{k+1, k+2\}$.

  Consequently, the sub-matrices $M_i H_{k+1}$ and $M_i H_{k+2}$ can be expressed as rank-1 matrices constructed from row vectors $\vu_i, \vv_i \in \Fq^{1 \times 2}$. Specifically, there exist scalars $a,b,c,d \in \Fq$ such that
  $M_{i} H_{k+1}=
  \begin{bmatrix}
    a \vu_i \\
    b \vu_i
  \end{bmatrix}$ and $ M_{i} H_{k+2}=
  \begin{bmatrix}
    c \vv_i \\
    d \vv_i
  \end{bmatrix},$
  where $\vu_i$ and $\vv_i$ are nonzero.

  Note that the matrix $[H_{k+1} ~ H_{k+2}]$ forms the identity matrix $I_4$. We can reconstruct $M_i$ as follows:
  \begin{align*}
    M_i = M_i[ H_{k+1} ~ H_{k+2}] = \Lambda
    \begin{bmatrix}
      \vu_i &  \\
       & \vv_i
    \end{bmatrix},
     \text{~where~} \Lambda= \begin{bmatrix} a & c \\ b & d \end{bmatrix}.
  \end{align*}
  Finally, since $M_i$ has rank $2$ and the block-diagonal matrix containing $\vu_i, \vv_i$ has rank $2$, the matrix $\Lambda$ must also be invertible.
\end{proof}

Let $V$ be a 2-length row vector or $2\times 2$ matrix over $\Fq$. Denote by $\Angle{V}$ the subspace spanned by $V$: $\Angle{V} = \{ \theta \cdot V : \theta\in \Fq\}$.
\begin{lemma}\label{lem:<>=<>}
    For $i\in[t]$ and $j\in[k]$, $j\notin \cR_i$ if and only if $\left\langle\vu_i\right\rangle=\left\langle\vv_i W_j\right\rangle$.
\end{lemma}
\begin{proof}
    By the definition of $\cR_i$, the condition $j\notin \cR_i$ is equivalent to $\rk{M_iH_j}=1$. Recall that $H_j=\begin{bmatrix}
        I_2\\
        W_j
    \end{bmatrix}$ for $j\in[k]$. By Lemma \ref{lem:bw-Fi},  $\rk{M_iH_j}=1$ if and only if  $\left\langle\vu_i\right\rangle=\left\langle\vv_i W_j\right\rangle$.
\end{proof}
The following theorem summarizes key properties of $\text{PGL}_2(\Fq)$ regarding transitivity, derived from \cite[Thm 9.8, 9.48]{rotman2012introduction}.

\begin{theorem}[Sharp 3-transitivity of $\text{PGL}_2(\Fq)$]\label{thm:sharp3}
  For any two ordered triples of pairwise distinct points
    \[
      \left(\Angle{\bm{u}_1},\Angle{\bm{u}_2},\Angle{\bm{u}_3}\right), \left(\Angle{\bm{v}_1},\Angle{\bm{v}_2},\Angle{\bm{v}_3}\right) \in (\mathbb{P}^1(\Fq))^3
    \]
    with $\Angle{\bm{u}_i} \neq \Angle{\bm{u}_j}$ and $\Angle{\bm{v}_i} \neq \Angle{\bm{v}_j}$ for $i \neq j$, there exists a unique
    projective linear transformation $\Angle{W}\in PGL_2(\Fq)$ such that
    $
      \Angle{\bm{u}_1} = \Angle{\bm{v}_1W}, 
      \Angle{\bm{u}_2} = \Angle{\bm{v}_2W}, 
      \Angle{\bm{u}_3} = \Angle{\bm{v}_3W}.
    $
\end{theorem}

Using the above results, we establish the following key lemma on the repair bandwidth.
\begin{lemma}\label{lem:bw-structure result}
  Among the sets $\cR_1, \cR_2, \ldots, \cR_t$, there are no more than $4$ that are distinct. Furthermore, if the number of distinct sets is $4$, then it must be the case that $t = k$.
\end{lemma}
\begin{proof}
  We prove the first assertion by contradiction. Suppose that there exist $5$ distinct sets in $\{\cR_{i}:i\in[t]\}$. Let $\cR_{a}\prec \cR_{b}\prec \cR_{c}\prec \cR_{d}\prec \cR_{e}$ denote these five sets in increasing order.
  
  Lemma~\ref{lem:bw-SiSj} implies that $d,e\notin \cR_{a}$, $ \cR_{b}$ and $\cR_{c}$. Consequently, by Lemma~\ref{lem:<>=<>}, we have
    \begin{align*}
      \Angle{\bm{u}_a} = \Angle{\bm{v}_aW_d}, \quad
      \Angle{\bm{u}_b} = \Angle{\bm{v}_bW_d}, \quad
      \Angle{\bm{u}_c} = \Angle{\bm{v}_cW_d}, \\
      \Angle{\bm{u}_a} = \Angle{\bm{v}_aW_e}, \quad
      \Angle{\bm{u}_b} = \Angle{\bm{v}_bW_e}, \quad
      \Angle{\bm{u}_c} = \Angle{\bm{v}_cW_e}.
    \end{align*}
    
  We assert that the projective points $\{\Angle{\vu_a}, \Angle{\vu_b}, \Angle{\vu_c}\}$ are all distinct, and likewise $\{\Angle{\vv_a}, \Angle{\vv_b}, \Angle{\vv_c}\}$ are all distinct. Otherwise, two of the sets among $\cR_a, \cR_b, \cR_c$ would coincide, contradicting the strict ordering $\cR_a \prec \cR_b \prec \cR_c$.
  
  Applying Theorem~\ref{thm:sharp3}, $\Angle{W_d} = \Angle{W_e}$. Recall that $d \in \cR_{d}$ (by Remark~\ref{rem:i in S}), so by Lemma~\ref{lem:<>=<>} we obtain $\Angle{\vu_d} \neq \Angle{\vv_d W_d}=\Angle{\vv_d W_e}$. It follows that $e \in \cR_d$. However, the ordering $\cR_d \prec \cR_e$ implies $e \notin \cR_d$ (by Lemma~\ref{lem:bw-SiSj}). This contradiction proves that $|\{\cR_{i}:i\in[t]\}| \leq 4$.

  Next, assume $\left|\{\cR_{i}:i\in[t]\}\right|= 4$.  
  Let $a,b,c,b\in[t]$ such that $ \cR_{a}\prec \cR_{b}\prec \cR_{c}\prec \cR_{d}$. Recall that $ \cR_{a}\prec \cR_{b}\prec \cR_{c}\prec \cR_{d}\prec \cR_{t+1}$. 
  Suppose, for the sake of contradiction, that $t < k$. 
  Since $t+1 \leq k$, the matrix $W_{t+1}$ exists. We can now repeat the argument from the first part of the proof. It leads to the conclusion that $t+1 \in \cR_d$.
  However, $\cR_d \prec \cR_{t+1}$ implies $t+1 \notin \cR_d$. This contradiction implies that the assumption $t < k$ is false. Therefore, we must have $t=k$.
\end{proof}

Now, we are ready to prove the lower bound on the avg-min repair bandwidth in Theorem~\ref{thm:avg-min-bw}.

\begin{proof}[Proof of Theorem~\ref{thm:avg-min-bw}]
  Recall that the avg-min repair bandwidth is given in \eqref{eq:sum_bound} as
  \(\bar\beta(\cC)=k+\frac{1}{n}{\sum_{ i\in [n]}|\cR_{i}|}.\)  Since $\cR_{t+1} = \dots = \cR_{k+1}$ and $|\cR_{k+1}| \leq |\cR_{k+2}|$, we can bound the sum as:
  $$\sum_{i\in[n]}|\cR_{i}| \geq \sum_{i\in[t]}|\cR_{i}| + (k+2-t)|\cR_{k+1}|. $$

    We analyze the lower bound in two cases based on the value of $t$.

    \textbf{Case 1: $t < k$.}
  According to Lemma~\ref{lem:bw-structure result}, the family $\{\cR_{i}:i\in[t]\}$ contains at most 3 distinct sets, denoted by $\{\cS_1, \cS_2, \cS_3\}$. 
  Let $s_j = |\{ i \in [t] : \cR_i = \cS_j \}|$ be the multiplicity of each set. It follows that $s_1 + s_2 + s_3 \geq t$. Crucially, since $i \in \cR_i$ for all $i$, the set $\cS_j$ must contain all indices $i$ such that $\cR_i = \cS_j$. Thus, $|\cS_j| \geq s_j$. Applying the mean inequality, we have:
  \begin{align*}
    \sum_{i\in[t]}|\cR_{i}| &= \sum_{j=1}^3 s_j |\cS_j| \geq \sum_{j=1}^3 s_j^2 \geq \frac{1}{3}\left(\sum_{j=1}^3 s_j\right)^2 \geq \frac{t^2}{3}.
  \end{align*}
   Furthermore, since $\cR_i \prec \cR_{k+1}$ for all $i \in [t]$, we have $|\cR_{k+1}| \ge \max_{j\in[3]} |\cS_j| \ge \max_{j\in[3]} s_j \ge t/3$. Combined with Lemma~\ref{lem:structure_t}, we obtain 
  \(
    |\cR_{k+1}| \geq \max\left\{\frac{t}{3}, k+1-t\right\}.
  \)
   Since $|\cR_{k+1}|\geq \max\{|\cS_1|,|\cS_2|,|\cS_3|\}$ and $|\cR_{k+1}|\geq k+1-t$, it follows that $|\cR_{k+1}|\geq \max\{\frac{t}{3},k+1-t\}$. Substituting these into \eqref{eq:sum_bound}, the total sum is bounded by:
  \begin{align*}
    \sum_{i\in[n]}|\cR_{i}| & \geq \sum_{i\in[t]}|\cR_{i}|+(k+2-t)|\cR_{k+1}|            \\
    & \geq \frac{t^2}{3}+ (k+2-t)\max\{\frac{t}{3},k+1-t\} \\
    & \geq \frac{1}{4}(k+1)(k+2).
  \end{align*}

  \textbf{Case 2: $t = k$.}
  By Lemma~\ref{lem:bw-structure result}, we have $|\{\cR_{i}:i\in[k]\}|\leq 4$. there are at most 4 distinct sets in $\{\cR_{i}:i\in[k]\}$. Similar to the previous case, let $s_j$ be the multiplicity of the distinct sets for $j\in[4]$. We have $\sum_{j=1}^4 s_j \geq k$ and $|\cS_j| \ge s_j$. Then, 
  \begin{align*}
    \sum_{i\in[k]}|\cR_{i}| \geq \sum_{j=1}^4 s_j^2 \geq \frac{1}{4}\left(\sum_{j=1}^4 s_j\right)^2 \geq \frac{k^2}{4}.
  \end{align*}
  Also, $|\cR_{k+1}| \geq \max_{j\in[4]} |\cS_j| \geq \frac{k}{4}$. Thus, 
  \begin{align*}
    \sum_{i\in[n]}|\cR_{i}| \geq \frac{k^2}{4} + 2\left(\frac{k}{4}\right) = \frac{k^2+2k}{4} = \frac{k(k+2)}{4}.
  \end{align*}

  By substituting the two cases into \eqref{eq:sum_bound}, we can complete the proof.
 \end{proof}

\subsection{The Lower Bound on Avg-Min Repair I/O}

Similar to $\sR(M)$ and $\sM_i$, for an $2\times 4$ repair matrix $M$, we define the following two sets
\begin{align*}
    &\cN(M) :=\{i\in [n] : \nz{MH_i} = 2\},\\
    &\cM'_i :=\{M\in \Fq^{\ell\times r\ell}:\rk{M}=\ell,~i\in \cN(M)\}.
\end{align*}
According to $d(M)=k+1$, the repair I/O corresponding to the repair matrix $M$ is
\begin{equation*}
    \tO(M) = \sum_{j\in [n]}\nz{MH_j}-2 = k + |\cN(M)|.
\end{equation*}

For each node $C_i$, let $\cN_i$ be the minimum set in $\{\cN(M):M\in\cM'_i\}$ under the total order $\prec$. Let $M'_i\in\cM'_i$ be a repair matrix such that $\cN(M'_i)=\cN_i$.

Since $\cR(M)\subseteq \cN(M)$, we have $\cM_i\subseteq\cM'_i$.
Therefore,
\begin{align*}
    \bar\gamma(\cC) & =\sum_{i\in[n]} \min_{M\in\cM_i}\{\tO(M)\}      \\& \geq \sum_{i\in[n]} \min_{M\in\cM'_i}\{\tO(M)\} =\sum_{i\in[n]} \tO(M'_i)
    \\& =k+\sum_{i\in[n]} |\cN_i|.
\end{align*}

\begin{remark}
    Note that, for every node $C_i$, $i\in[n]$, we always have $i\in\cN_i$.
\end{remark}
\begin{lemma}\label{lem:io-SiSj}
    For any two distinct nodes $C_i$ and $C_j, i,j\in[n]$, if $\cN_i\prec \cN_j$, then $j\notin \cN_i$.
\end{lemma}
\begin{proof}
    Suppose that $j\in \cN_i$, we have $M'_i\in \cM'_j$. In either case, this implies that $\cN_j \preceq \cN_i$. Since  $\cN_j$ is minimized under $\prec$. This contradicts the assumption that $\cN_i\prec \cN_j$. Therefore, we conclude that $j\notin \cN_i$.
\end{proof}

Similarly, we assume $\cN_1 \preceq \cN_2 \preceq \dots \preceq \cN_n$ for the rest of this subsection.

\begin{lemma}\label{lem:io-t}
    There exists an integer $i\in[k]$ such that $\cN_i\prec \cN_{k+1}$. Let
    $
        t := \max\{i\in[k]: \cN_i \prec \cN_{k+1}\}.
    $
    Then, $|\cN_{k+1}|\geq k+1-t$.
\end{lemma}

\begin{proof}
    First, we verify the existence of such an integer $i$. Suppose, for the sake of contradiction, that no such $i$ exists. This would imply $\cN_{1}=\dots=\cN_{k+1}$. Such a condition would result in a repair I/O strictly greater than $2k$, contradicting Lemma~\ref{lem:d=k+1}. Thus, the set $\{i\in[k]: \cN_i \prec \cN_{k+1}\}$ is non-empty, and $t$ is well-defined.

    By the definition of $t$, we have $\cN_{t+1} =\dots=\cN_{k+1}$. Remark~\ref{rem:i in S} implies that 
    $
        \{t+1, \dots, k+1\} \subseteq \cN_{k+1}.
    $
    Consequently, $|\cN_{k+1}| \geq |\{t+1, \dots, k+1\}| = k+1-t$.
\end{proof}

The following lemma demonstrates the structure of the repair matrices for the first $t$ nodes.
\begin{lemma}\label{lem:Fi'}
    For each $i\in[t]$, the repair matrix $M'_{i}$ can be decomposed as
    $$ M'_{i}=\Lambda\begin{bmatrix}
            \vu_i &          \\
                     & \vv_i
        \end{bmatrix},$$
    where $\Lambda\in \Fq^{2\times 2}$ is an invertible matrix.
    Moreover, we have
    \begin{equation*}
        \{\langle\vu_i\rangle,\langle\vv_i\rangle:~i\in[t]\} \subseteq\{\langle \ve_1\rangle,\langle \ve_2\rangle\},
    \end{equation*}
    with $\ve_1=\begin{pmatrix}
            1 & 0
        \end{pmatrix}$ and $\ve_2=\begin{pmatrix}
            0 & 1
        \end{pmatrix}$.
\end{lemma}
\begin{proof}
    Since $\cN_i\preceq \cN_t\prec \cN_{k+1} \preceq \cN_{k+2}$ for each $i\in[t]$, Lemma~\ref{lem:io-SiSj} implies that $k+1,k+2\notin \cN_i$.

    By the definition of $\cN_i$, we have $\rk{M'_i H_j} \leq \nz{M'_i H_j}\leq 1$ for $j \in \{k+1,k+2\}$. On the other hand, since the repair degree is $k+1$, we have $\rk{M'_i H_j} = \nz{M'_i H_j}= 1$ for $j \in \{k+1,k+2\}$.

    Consequently, the sub-matrices $M'_i H_{k+1}$ and $M'_i H_{k+2}$ can be expressed as rank-1 matrices constructed from non-zero row vectors $\vu_i, \vv_i \in \Fq^{1 \times 2}$. Specifically, there exist scalars $a,b,c,d \in \Fq$ such that
    $M'_{i} H_{k+1}=
    \begin{bmatrix}
    a \vu_i \\
    b \vu_i
    \end{bmatrix}$ and $ M'_{i} H_{k+2}=
    \begin{bmatrix}
    c \vv_i \\
    d \vv_i
    \end{bmatrix},$
    where $\vu_i$ and $\vv_i$ are nonzero. Since $\nz{M'_i H_{k+1}}=\nz{M'_i H_{k+2}}=1$, we have
    \begin{equation*}
        \{\langle\vu_i\rangle,\langle\vv_i\rangle:~i\in[t]\} \subseteq\{\langle \ve_1\rangle,\langle \ve_2\rangle\},
    \end{equation*}
    where $\ve_1=\begin{pmatrix}
            1 & 0
        \end{pmatrix}$ and $\ve_2=\begin{pmatrix}
            0 & 1
        \end{pmatrix}$.

    Note that the matrix $[H_{k+1} ~ H_{k+2}]$ forms the identity matrix $I_4$. We can reconstruct $M'_i$ as follows:
    \begin{align*}
    M'_i = M'_i[ H_{k+1} ~ H_{k+2}] = \Lambda
    \begin{bmatrix}
      \vu_i &  \\
       & \vv_i
    \end{bmatrix},
     \text{~where~} \Lambda= \begin{bmatrix} a & c \\ b & d \end{bmatrix}.
    \end{align*}
    
    Finally, since $M'_i$ has rank $2$ and the block-diagonal matrix containing $\vu_i, \vv_i$ has rank $2$, the matrix $\Lambda$ must also be invertible.
\end{proof}
      
    So there are only four choices of $(\langle \vu_i \rangle, \langle \vv_i \rangle)$ for all $i \in [t]$: $(\langle \ve_1 \rangle, \langle \ve_1 \rangle)$ ,$(\langle \ve_1 \rangle, \langle \ve_2 \rangle)$ ,$(\langle \ve_2 \rangle, \langle \ve_1 \rangle)$ ,$(\langle \ve_2 \rangle, \langle \ve_2 \rangle)$.

Since the two repair matrices $M'_{i}$ and $\Lambda^{-1}
    M'_{i}$ behave the same during the repair
process, we always assume that $$M'_{i}=\begin{bmatrix}
        \vu_i &          \\
                 & \vv_i
    \end{bmatrix} \text{~for~all~} i\in[t].$$

\begin{lemma}\label{lem:io-ui-vi}
    For any $i, j \in [t]$, 
    \begin{enumerate}
        \item if $\langle \vu_i \rangle = \langle \vu_j \rangle$ and $\langle \vv_i \rangle = \langle \vv_j \rangle$, then $\cN_i = \cN_j$;
       \item if $\langle \vu_i \rangle=\langle \vu_j \rangle$, $\langle \vv_i
                  \rangle \neq \langle \vv_j \rangle$ or $\langle \vu_i \rangle \neq
                  \langle \vu_j \rangle$, $\langle \vv_i \rangle=\langle \vv_j \rangle$,
              then $ \cN_i\cup \cN_j=[k]$.
    \end{enumerate}
\end{lemma}
\begin{proof}
    We only prove the second statement. Suppose for the sake of contradiction that there exists an element $a \in [k]$ such that $a \notin \cN_i \cup \cN_j$. 
    By definition, we have $\nz{M'_i H_a} = \nz{M'_j H_a} = 1$, which means that 
    $$ \RK{\begin{bmatrix} \vu_i \mA_a \\ \vv_i \mB_a \end{bmatrix}} = \RK{\begin{bmatrix} \vu_j \mA_a \\ \vv_j \mB_a \end{bmatrix}} = 1. $$
    This implies that $\langle \vu_i \mA_a \rangle = \langle \vv_i \mB_a \rangle$ and $\langle \vu_j \mA_a \rangle = \langle \vv_j \mB_a \rangle$.  
    
    By symmetry, we may assume without loss of generality that $\langle \vu_i \rangle = \langle \vu_j \rangle$ but $\langle \vv_i \rangle \neq \langle \vv_j \rangle$. The assumption $\langle \vu_i \rangle = \langle \vu_j \rangle$ directly gives $\langle \vu_i \mA_a \rangle = \langle \vu_j \mA_a \rangle$. Consequently, we obtain the following chain of equalities:
    $$ \langle \vv_i \rangle = \langle \vu_i \mA_a \mB_a^{-1} \rangle = \langle \vu_j \mA_a \mB_a^{-1} \rangle = \langle \vv_j \rangle. $$
    This contradicts the premise that $\langle \vv_i \rangle \neq \langle \vv_j \rangle$. A symmetric contradiction arises for the other case. Thus, $\cN_i \cup \cN_j = [k]$.
\end{proof}

Using the above lemmas, we establish the following properties of $\cN_i$ where $i\in[t]$.

\begin{lemma}\label{lem:io-structure result}
    For the sequence of sets $(\cN_i)_{i \in [t]}$, the following hold:
    \begin{enumerate}
        \item The number of distinct sets is bounded by $|\{\cN_i : i \in [t]\}| \leq 3$. Furthermore, if there are exactly $3$ distinct sets, then $t=k$.
        \item For any $a \in [t]$, there exists $b \in [t]$ such that $\cN_a \cup \cN_b = [t]$.    
    \end{enumerate}
\end{lemma}

\begin{proof}
    \begin{enumerate}      
        \item The vectors $\langle \vu_i \rangle$ and $\langle \vv_i \rangle$ are chosen from $\{\langle \ve_1 \rangle, \langle \ve_2 \rangle\}$, yielding at most $4$ possible combinations for the pairs $(\langle \vu_i \rangle, \langle \vv_i \rangle)$. By Lemma~\ref{lem:io-ui-vi} (1), this gives an initial bound of $|\{\cN_i : i \in [t]\}| \leq 4$.
        
        \textbf{Proof of $t=k$ for $3$ distinct sets:} If there are exactly $3$ distinct sets, any subset of $3$ pairs chosen from the $4$ possible configurations must contain two pairs that share exactly one coordinate. By Lemma~\ref{lem:io-ui-vi} (2), the union of these two sets is $[k]$. Since their union must be contained within $[t]$ and we are given $t \leq k$, this strict inclusion forces $t=k$.
        
        \textbf{Proof of the bound $|\{\cN_i: i \in [t]\}| \leq 3$:} Suppose, for contradiction, that exactly $4$ distinct sets exist, meaning all four possible coordinate pairs are realized. Let these four sets be denoted by $\cS_1 \prec \cS_2 \prec \cS_3 \prec \cS_4 = \cN_t$.
        Observe that among the first three strictly non-maximal sets $\{\cS_1, \cS_2, \cS_3\}$, there must exist a pair $\cS_x$ and $\cS_y$ (with $x < y \leq 3$) whose corresponding coordinate pairs share exactly one dimension. By Lemma~\ref{lem:io-ui-vi} (2), their union is $\cS_x \cup \cS_y = [k]$. 
        Since having $\ge 3$ distinct sets implies $t=k$ (as proven above), we have $\cS_x \cup \cS_y = [t]$. However, by Lemma~\ref{lem:io-SiSj}, the relation $\cS_x \preceq \cS_y\prec \cS_4$ imply a contradiction that $t \notin \cS_x\cup \cS_y=[t]$. Consequently, we conclude $|\{\cN_i : i \in [t]\}| \leq 3$.
        
        \item Fix an element $a \in [t]$. We consider two cases for the remaining elements:
        
        If there exists $b \in [t]$ such that exactly one of the identities $\langle \vu_b \rangle = \langle \vu_a \rangle$ or $\langle \vv_b \rangle = \langle \vv_a \rangle$ holds, then by Lemma~\ref{lem:io-ui-vi} (2), $\cN_a \cup \cN_b = [k]$. This implies $t=k$ as shown in Part 1. Thus, $\cN_a \cup \cN_b = [t]$, and the proof is complete.
        
        Otherwise, for all $b \in [t]$, the pairs $(\langle \vu_b \rangle, \langle \vv_b \rangle)$ must either match $(\langle \vu_a \rangle, \langle \vv_a \rangle)$ completely, or differ in both coordinates. This restricts the set of all pairs $\{(\langle \vu_i \rangle, \langle \vv_i \rangle) : i \in [t]\}$ to at most two distinct configurations. By Lemma~\ref{lem:io-ui-vi} (1), this implies $|\{\cN_i : i \in [t]\}| \leq 2$. Since $\bigcup_{i \in [t]} \cN_i = [t]$, we can always find some $b \in [t]$ such that $\cN_a \cup \cN_b = [t]$.
    \end{enumerate}
\end{proof}

The proof of the lower bound on the avg-min repair I/O in Theorem~\ref{thm:avg-min-io} is as follows.
\begin{proof}[Proof of Theorem~\ref{thm:avg-min-io}]
    Recall that the average metric is bounded by
    $$ \bar\gamma(\cC) \geq \frac{1}{n} \sum_{i\in [n]} \tO(M'_i) = \frac{1}{n} \sum_{i\in [n]} |\cN_i| + k. $$  
    Given the ordering $|\cN_{k+1}| \leq |\cN_n|$ and the uniform trailing sets $\cN_{t+1} = \dots = \cN_{k+1}$, we can partition and bound the total sum as:
    $$ \sum_{i\in[n]}|\cN_i| \geq \sum_{i\in[t]}|\cN_i| + (k+2-t)|\cN_{k+1}|. $$

    According to Lemma~\ref{lem:io-structure result}, the number of distinct sets within the prefix is strictly bounded by $|\{\cN_i: i \in [t]\}| \leq 3$. We will analyze the partial sum $\sum_{i\in[t]}|\cN_i|$ by exhaustively discussing the three possible cases.

    \vspace{0.5em}
    \noindent\textbf{Case 1:} $|\{\cN_i : i \in [t]\}| = 1$. 
    This implies total uniformity: $\cN_1 = \dots = \cN_t$. Since every index satisfies $i \in \cN_i$, the single unique set must contain all $i \in [t]$, meaning its size is at least $t$. Thus, 
    $$ \sum_{i\in[t]}|\cN_i| \geq t \cdot t = t^2. $$

    \noindent\textbf{Case 2:} $|\{\cN_i : i \in [t]\}| = 2$. 
    Assume the two distinct sets are $\cS_1 \prec \cS_2$. We partition the index set $[t]$ into $\cT_1 = \{i \in [t] : \cN_i = \cS_1\}$ and $\cT_2 = \{i \in [t] : \cN_i = \cS_2\}$. Since $i \in \cN_i$, it follows naturally that $\cT_1 \subseteq \cS_1$ and $\cT_2 \subseteq \cS_2$, with $|\cT_1| + |\cT_2| = t$. 
    Applying the Cauchy-Schwarz inequality, we compute:
    \begin{align*}
        \sum_{i\in[t]}|\cN_i| &= |\cT_1||\cS_1| + |\cT_2||\cS_2| \\
        &\geq |\cT_1|^2 + |\cT_2|^2 \\
        &\geq \frac{(|\cT_1| + |\cT_2|)^2}{2} = \frac{t^2}{2}.
    \end{align*}

    \noindent\textbf{Case 3:} $|\{\cN_i : i \in [t]\}| = 3$. 
    By the property (1) in Lemma~\ref{lem:io-structure result}, having exactly 3 distinct sets forces $t = k$. 
    Let the sets be strictly ordered as $\cS_1 \prec \cS_2 \prec \cS_3$. We partition the index space $[k]$ into $\cT_j = \{i \in [k] : \cN_i = \cS_j\}$ for $j \in \{1, 2, 3\}$. This guarantees $\cT_j \subseteq \cS_j$ and $\sum_{j=1}^3 |\cT_j| = k$. 
    
    By Lemma~\ref{lem:io-SiSj}, the strictly smaller set $\cS_1$ cannot contain indices belonging to the larger sets, yielding $(\cT_2 \cup \cT_3) \cap \cS_1 = \emptyset$. Since $\cT_1 \subseteq \cS_1$ and the partitions form $[k]$, this completely restricts $\cS_1 = \cT_1$. 
    Furthermore, Lemma~\ref{lem:io-structure result} dictates the maximal coverages $\cS_1 \cup \cS_3 = [k]$ and $\cS_2 \cup \cS_3 = [k]$, which implies $[k] \setminus \cS_3 \subseteq \cS_1 \cap \cS_2$. Because $\cS_1 = \cT_1$ and the index partitions are mutually disjoint, we deduce $([k] \setminus \cS_3) \cap \cT_2 \subseteq \cT_1 \cap \cT_2 = \emptyset$. 
    Therefore, utilizing $\cT_2 \subseteq \cS_2$, we have:
    \begin{equation*}
        |\cS_2| \geq |([k] \setminus \cS_3) \cup \cT_2| = |[k] \setminus \cS_3| + |\cT_2|.
    \end{equation*}

    Let $c = |[k] \setminus \cS_3|$. The inclusion $[k] \setminus \cS_3 \subseteq \cS_1 = \cT_1$ naturally bounds $c$ within $0 \leq c \leq |\cT_1|$. Noting that $|\cS_3| = k - c$, we can expand the sum:
    \begin{align*}
        \sum_{i\in[k]}|\cN_i| &= |\cT_1||\cS_1| + |\cT_2||\cS_2| + |\cT_3||\cS_3| \\
        &\geq |\cT_1|^2 + |\cT_2|(|\cT_2| + c) + |\cT_3|(k - c) \\
        &= |\cT_1|^2 + |\cT_2|^2 + |\cT_3|k + (|\cT_2| - |\cT_3|)c.
    \end{align*}
    Notice that the expression is a linear function with respect to $c \in [0, |\cT_1|]$. Its minimum must therefore occur at one of the boundary points, $c = 0$ or $c = |\cT_1|$. Checking both extremes yields:
    \begin{align*}
        \sum_{i\in[k]}|\cN_i| &\geq \min \Big\{ |\cT_1|^2 + |\cT_2|^2 + |\cT_3|k, \; |\cT_1|^2 + |\cT_3|^2 + |\cT_2|k \Big\} \\
        &\geq \min \left\{ \frac{(|\cT_1| + |\cT_2|)^2}{2} + |\cT_3|k, \; \frac{(|\cT_1| + |\cT_3|)^2}{2} + |\cT_2|k \right\}.
    \end{align*}
    By substituting $|\cT_1| + |\cT_2| = k - |\cT_3|$ and $|\cT_1| + |\cT_3| = k - |\cT_2|$, the expression beautifully simplifies to:
    \begin{align*}
        \sum_{i\in[k]}|\cN_i| &\geq \min \left\{ \frac{k^2 + |\cT_3|^2}{2}, \; \frac{k^2 + |\cT_2|^2}{2} \right\} \geq \frac{k^2 + 1}{2}.
    \end{align*}

    \vspace{0.5em}
    Synthesizing the discussions from Cases 1, 2, and 3, we establish the universal lower bound $\sum_{i\in [t]}|\cN_i| \geq t^2/2$. 
    Given the size properties $|\cN_{k+1}| \geq |\cN_t| \geq t/2$ and $|\cN_{k+1}| \geq k + 1 - t$, we can lower-bound the total sum:
    \begin{align*}
        \sum_{i\in[n]}|\cN_i| &\geq \sum_{i\in[t]}|\cN_i| + (k+2-t)|\cN_{k+1}| \\
        &\geq \frac{t^2}{2} + (k+2-t)\max\left\{ \frac{t}{2}, k+1-t \right\} \\
        &\geq \frac{(k+1)(k+2)}{3}.
    \end{align*}
    Consequently, dividing by $n$ and adding $k$, we obtain the final bound for the average metric:
    $$ \bar\gamma(\cC) = \frac{1}{n}\sum_{ i\in [n]}|\cN_i| + k \geq \frac{4k+1}{3}. $$
\end{proof}

\section{\texorpdfstring{$(k+2,k,2)$}{} MDS Array Codes with Optimal Repair}\label{sec:construction}

We construct two classes of explicit $(n = k +2, k, \ell = 2)$ MDS array codes. The first class of MDS array codes achieves the lower bounds on the avg-min repair bandwidth in Theorem~\ref{thm:avg-min-bw} and max-min repair bandwidth in Corollary~\ref{coro:max-min-bw} asymptotically. The second class of MDS array codes achieves the lower bounds on the avg-min repair I/O in Theorem~\ref{thm:avg-min-io} and the max-min repair I/O in Corollary~\ref{coro:max-min-io}. These two classes of codes show that all the proposed lower bounds are tight.

Recall that an $(n = k+2, k, \ell = 2)$ MDS array code comprises $n$ nodes, denoted by $C_1, C_2, \dots, C_n$, where each node is a column vector of length $\ell = 2$. An MDS array code can be defined by the parity check sub-matrices $H_i$, $i\in [n]$, and the parity-check equations~\eqref{eq:pc} as described in Section~\ref{sec:linear-repair}. The defined code is MDS if and only if any two parity check sub-matrices $H_i, H_j$ (where $1\le i < j \le n$) can form an invertible matrix. The linear repair process for each node $C_i$ can be described by a $2\times 4$ repair matrix $M_i$ (where $i\in [n]$). The repair bandwidth and repair I/O for each node $C_i$ can be calculated using the formulas~\eqref{eq:BW} and \eqref{eq:IO}, respectively.

\subsection{\texorpdfstring{$(k+2,k,2)$}{} MDS Array Codes with Optimal Bandwidth}

\noindent\textbf{Code construction ($\cC_1$).}
Let $\Fq$ be the finite field of size $q\ge n+3$, and $\alpha$ a primitive element of $\Fq$.
Then, for each $i, 0\le i\le n+2$, we define an element $\lambda_i$ and the corresponding Vandermonde column vector as
\[
    \lambda_i = \alpha^i \text{~and~} \vlmd_i=[1, \lambda_i]^T.
\]
Any two distinct $\vlmd_i$ and $\vlmd_j$ (where $0\le i, j\le n +2$) are linearly independent.

We evenly partition the $n$ nodes $C_1, C_2, \dots, C_n$ into four node groups with index sets are denoted by $\sG_1, \sG_2, \sG_3, \sG_4$. Each of the first $n \MOD 4$ groups contains $\ceil{n/4}$ nodes, while each of the remaining groups contains $\floor{n/4}$ nodes. For example, if $n = 6$, then $\sG_1 = \{1, 2\}$, $\sG_2 = \{3, 4\}$, $\sG_3 = \{5\}$, and $\sG_4 = \{6\}$.

We give the parity check sub-matrix $H_i$ and repair matrix $M_i$ for each node $C_i$ (where $1\le i\le n$) in Table~\ref{tab:cons-bw}. The following lemma shows that the resulting code $\cC_1$ is MDS.
\begin{table}[ht]
    \centering
    \begin{tabular}{|c|c|c|}
        \hline
        $i$      & $M_i$  & $H_i$                                                                                                                                                                 \\\hline &&\\
        $i\in \sG_1$    & $\begin{bmatrix} {\ \ }1 &{\ \ }0 &{\ \ }0 &{\ \ }0 \\ {\ \ }0 &{\ \ }1 &{\ \ }0 &{\ \ }0\end{bmatrix}$ & $\begin{bmatrix}\vlmd_{i-1} &-\vlmd_i\\ \bm 0 &{\ \ }\vlmd_i \end{bmatrix}$ 
        \\&&\\\hline&&\\
        $i\in \sG_2$  & $\begin{bmatrix} {\ \ }0 &{\ \ }0 &{\ \ }1 &{\ \ }0 \\ {\ \ }0 &{\ \ }0 &{\ \ }0 &{\ \ }1\end{bmatrix}$ & $\begin{bmatrix}{\ \ }\vlmd_{i} &\bm 0\\ -\vlmd_i &\vlmd_{i+1} \end{bmatrix}$ 
        \\&&\\\hline&&\\
        $i\in \sG_3$  & $\begin{bmatrix} {\ \ }1 &{\ \ }0 &{\ \ }1 &{\ \ }0 \\ {\ \ }0 &{\ \ }1 &{\ \ }0 &{\ \ }1\end{bmatrix}$ & $\begin{bmatrix}{\ \ }\vlmd_{i} &\bm 0\\ \bm 0 &\vlmd_{i+2} \end{bmatrix}$
        \\&&\\\hline&&\\
        $i\in \sG_4$   & $\begin{bmatrix} -1 &{\ \ }0 &{\ \ }\alpha &{\ \ }0 \\ {\ \ }0 &{\ \ }\alpha &{\ \ }0 &-1\end{bmatrix}$ & $\begin{bmatrix}\vlmd_{i+2} &\bm 0\\ \bm 0 &\vlmd_{i+2} \end{bmatrix}$
        \\&&\\\hline
    \end{tabular}
    \caption{Repair matrix $M_i$ and parity check sub-matrix $H_i$ for each node $C_i$ (where $1\le i\le n$) of code $\cC_1$.}
    \label{tab:cons-bw}
\end{table}
\begin{table}[ht]
    \centering
    \begin{tabular}{|c|c|c|c|c|}
        \hline
        $\rk{M_iH_j}$ & $j\in \sG_1$ & $j\in \sG_2$ & $j\in \sG_3$ & $j\in \sG_4$ \\\hline
        $i\in \sG_1$    & 2          & 1          & 1          & 1          \\\hline
        $i\in \sG_2$    & 1          & 2          & 1          & 1          \\\hline
        $i\in \sG_3$    & 1          & 1          & 2          & 1          \\\hline
        $i\in \sG_4$    & 1          & 1          & 1          & 2          \\\hline
    \end{tabular}
    \caption{The value of $\rk{M_iH_j}$ for each $i,j\in [n]$ of code $\cC_1$.}
    \label{tab:cons-bw-rank}
\end{table}

\begin{lemma}\label{Lem:MDS-C1}
    The code $\cC_1$ defined by the parity check sub-matrices in Table~\ref{tab:cons-bw} is a $(k+2, k, 2)$ MDS array code.
\end{lemma}
\begin{proof}
    We only need to show that each square matrix $[H_i, H_j]$ is invertible for all distinct indices $i, j$ (where $1\le i<j\le n$).

    One key observation for the structure of the parity check sub-matrices in Table~\ref{tab:cons-bw} is that any square matrix $[H_i, H_j]$ can be transformed into a block triangular matrix by an invertible transformation, and the blocks on the main block diagonal are invertible, which implies that the square matrix $[H_i, H_j]$ is also invertible.

    If $i\in \sG_1$ and $j\in \sG_2$, we can add the second block row of $[H_i, H_j]$ to the first block row, and then permute the columns of the matrix to obtain a block triangular matrix. The procedure is as follows:
    \begin{equation*}
        \begin{aligned}
             & \left[\begin{array}{cc|cc}
                             \vlmd_{i-1} & -\vlmd_i      & {\ \ }\vlmd_j & \bm 0   \\
                             \bm 0   & {\ \ }\vlmd_i & -\vlmd_j      & \vlmd_{j+1}
                         \end{array}\right]           \\
            \xrightarrow{\text{row transformation}}
             & \left[\begin{array}{cccc}
                             \vlmd_{i-1} & {\ \ }\bm 0 & {\ \ }\bm 0 & \vlmd_{j+1} \\
                             \bm 0   & {\ \ }\vlmd_i   & -\vlmd_j        & \vlmd_{j+1}
                         \end{array}\right]       \\
            \xrightarrow{\text{column permutation}}
             & \left[\begin{array}{cc|cc}
                             \vlmd_{i-1} & \vlmd_{j+1} & {\ }\bm 0 & {\ \ }\bm 0 \\\hline
                             \bm 0   & \vlmd_{j+1} & {\ }\vlmd_i   & -\vlmd_{j}
                         \end{array}\right].
        \end{aligned}
    \end{equation*}
    Since the square matrix on the main block diagonal is invertible, the original matrix $[H_i, H_j]$ is also invertible. For other cases, we can directly permute the columns of the matrix to obtain a block triangular matrix.
\end{proof}

\noindent\textbf{Repair bandwidth.}
Recall the linear repair process of each $C_i$ driven by the repair matrix $M_i$ (where $i\in [n]$) described in Section~\ref{sec:linear-repair}. We can calculate the repair bandwidth for  $C_i$ by the formulas~\eqref{eq:BW}.

Firstly, one can directly check from Table~\ref{tab:cons-bw} that for each $i\in [n]$, $\rk{M_iH_i} = 2$. Thus, $M_i$ can define a linear repair process for node $C_i$. Then we summarize the values of $\rk{M_iH_j}$ for all $i,j\in [n]$ in Table~\ref{tab:cons-bw-rank}. Now we can calculate the repair bandwidth for each node $C_i$ (where $i\in [n]$) by the formula~\eqref{eq:BW}:
\begin{align*}
    \tB(M_i) & = \sum_{j=1}^n \rk{M_iH_j} - 2 \\
             & = 2|\sG_z| + (n-|\sG_z|) - 2       \\ & = k + |\sG_z|,
\end{align*}
where $i\in \sG_z$.
According to the group partition, we know that $|\sG_z| = \ceil{n/4}$ or $|\sG_z| = \floor{n/4}$. The max-min repair bandwidth for code $\cC_1$ is $$\beta(\cC_1) = k+\ceil{n/4} = \ceil{\frac{5k+2}{4}},$$ and the avg-min repair bandwidth for code $\cC_1$ is $$\bar{\beta}(\cC_1) = k + \floor{n/4} + \frac{1}{n}(n \MOD 4)\ceil{n/4}.$$ We can conclude that the max-min repair bandwidth of code $\cC_1$ achieves the lower bound in Corollary~\ref{coro:max-min-bw} when $k \MOD 4 = 1$ or $k \MOD 4 = 2$. In addition, the avg-min repair bandwidth of code $\cC_1$ achieves the lower bound in Theorem~\ref{thm:avg-min-bw} when $k$ goes to infinity.

\subsection{\texorpdfstring{$(k+2,k,2)$}{} MDS Array Codes with Optimal Repair I/O}

\noindent\textbf{Code construction ($\cC_2$).}
Let $\Fq$ be a finite field of size $q\ge n+1$, and $\lambda_0, \lambda_1, \dots, \lambda_{n}$ be $n+1$ distinct elements in $\Fq$. We also denote a Vandermonde column vector $\vlmd_i = [1, \lambda_i]^T$ for each $\lambda_i$ (where $0\le i\le n$).

We evenly partition the $n$ nodes $C_1, C_2, \dots, C_n$ into three node groups with index sets denoted by $\sG_1, \sG_2, \sG_3$. Each of the first $n \MOD 3$ groups contains $\ceil{n/3}$ nodes, while each of the remaining groups contains $\floor{n/3}$ nodes. For example, if $n = 8$, then $\sG_1 = \{1, 2, 3\}$, $\sG_2 = \{4, 5, 6\}$, $\sG_3 = \{7, 8\}$.

We give the parity check sub-matrix $H_i$ and repair matrix $M_i$ for each node $C_i$ (where $1\le i\le n$) in Table~\ref{tab:cons-io}. Lemma~\ref{Lem:MDS-C2} shows that the resulting code $\cC_2$ is MDS. We omit the proof of this lemma since it is similar to the proof of Lemma~\ref{Lem:MDS-C1}.

\begin{table}[ht]
    \centering
    \begin{tabular}{|c|c|c|}
        \hline
        $i$     & $M_i$     & $H_i$                                                                                                                                                               \\\hline &&\\
        $i\in \sG_1$   & $\begin{bmatrix} {\ \ }1 &{\ \ }0 &{\ \ }0 &{\ \ }0 \\ {\ \ }0 &{\ \ }1 &{\ \ }0 &{\ \ }0\end{bmatrix}$ & $\begin{bmatrix}\vlmd_{i-1} &-\vlmd_i\\ \bm 0 &{\ \ }\vlmd_i \end{bmatrix}$  
        \\&&\\\hline&&\\
        $i\in \sG_2$   & $\begin{bmatrix} {\ \ }0 &{\ \ }0 &{\ \ }1 &{\ \ }0 \\ {\ \ }0 &{\ \ }0 &{\ \ }0 &{\ \ }1\end{bmatrix}$ & $\begin{bmatrix}{\ \ }\vlmd_{i} &\bm 0\\ -\vlmd_i &\vlmd_{i-1} \end{bmatrix}$
        \\&&\\\hline&&\\
        $i\in \sG_3$  & $\begin{bmatrix} {\ \ }1 &{\ \ }0 &{\ \ }1 &{\ \ }0 \\ {\ \ }0 &{\ \ }1 &{\ \ }0 &{\ \ }1\end{bmatrix}$ & $\begin{bmatrix}{\ \ }\vlmd_{i} &\bm 0\\ \bm 0 &\vlmd_{i+1} \end{bmatrix}$
        \\&&\\\hline
    \end{tabular}
    \caption{Repair matrix $M_i$ and parity check sub-matrix $H_i$ for each node $\cC_i$ (where $1\le i\le n$) of code $\cC_2$.}
    \label{tab:cons-io}
\end{table}
\begin{table}[ht]
    \centering
    \begin{tabular}{|c|c|c|c|}
        \hline
        $\nz{M_iH_j}$ & $j\in \sG_1$ & $j\in \sG_2$ & $j\in \sG_3$ \\\hline
        $i\in \sG_1$    & 2          & 1          & 1          \\\hline
        $i\in \sG_2$    & 1          & 2          & 1          \\\hline
        $i\in \sG_3$    & 1          & 1          & 2          \\\hline
    \end{tabular}
    \caption{The value of $\nz{M_iH_j}$ for each $i,j\in [n]$ of code $\cC_2$.}
    \label{tab:cons-io-nz}
\end{table}

\begin{lemma}\label{Lem:MDS-C2}
    The code $\cC_2$ defined by the parity check sub-matrices in Table~\ref{tab:cons-io} is a $(k+2, k, 2)$ MDS array code.
\end{lemma}

\noindent\textbf{Repair I/O.}
Recall the linear repair process of each $C_i$ driven by the repair matrix $M_i$ (where $i\in [n]$) described in Section~\ref{sec:linear-repair}.
We can calculate the repair I/O for $C_i$ using the formula~\eqref{eq:IO}.

Firstly, from Table \ref{tab:cons-io}, one can directly check that for each $i\in [n]$, $\rk{M_iH_i} = 2$. Thus, $M_i$ can define a linear repair process for node $C_i$. Then we summarize the values of $\nz{M_iH_j}$ for all $i,j\in [n]$ in Table~\ref{tab:cons-io-nz}. Now we can calculate the repair I/O for each node $C_i$ (where $i\in [n]$) using the formula~\eqref{eq:IO}:
\begin{align*}
    \tO(M_i) & = \sum_{j=1}^n \nz{M_iH_j} - 2 \\
             & = 2|\sG_z| + (n-|\sG_z|) - 2       \\ & = k + |\sG_z|,
\end{align*}
where $i\in \sG_z$.
According to the group partition of code $\cC_2$, we know that $|\sG_z| = \ceil{n/3}$ or $|\sG_z| = \floor{n/3}$. The max-min repair I/O for code $\cC_1$ is $$\gamma(\cC_1) = k+\ceil{n/3} = \ceil{\frac{4k+2}{3}},$$ and the avg-min repair I/O for code $\cC_2$ is $$\bar \gamma(\cC_2) = k + \floor{n/3} + \frac{1}{n}(n \MOD 3)\ceil{n/3}.$$ We can conclude that the max-min repair I/O of code $\cC_2$ achieves the lower bound in Corollary~\ref{coro:max-min-io} when $k \MOD 3 = 0$ or $k \MOD 3 = 1$. In addition, the avg-min repair I/O of code $\cC_2$ achieves the lower bound in Theorem~\ref{thm:avg-min-io} when $k$ goes to infinity.


\section{Conclusion}\label{sec:conclusion}

We investigate lower bounds on the repair bandwidth and repair I/O for single-node failures in MDS array codes with $r=2$ and $\ell=2$.
Under these parameters, we map repair schemes to point sets on the projective line \(\mathbb{P}^1\). Leveraging the sharply 3-transitive action of \(\text{PGL}_2(\Fq)\), we derive a lower bound on the repair bandwidth and provide explicit MDS array code constructions that meet this bound. Using a similar approach, we also establish a lower bound for repair I/O and provide an optimal construction.

An interesting future research direction is to establish tight lower bounds for the repair bandwidth and repair I/O of the MDS array codes with general parameters $(n,k,\ell)$. Deriving such bounds for specific small sub-packetization levels (e.g., $\ell = 2, 4, 8$) is also valuable, as these values align with practical implementations in distributed storage systems.

\section*{Acknowledgment}

We would like to express our sincere gratitude to Dr. Huawei Wu for providing valuable suggestions during our discussions.
\bibliographystyle{ieeetr}
\bibliography{ref}

\end{document}